\documentclass[11pt,a4paper]{article}
\usepackage{amssymb} 
\usepackage{amsmath} 
\usepackage{amsthm} 
\usepackage{fullpage,enumerate}

\newtheorem{theorem}{Theorem}

\newtheorem{lemma}[theorem]{Lemma}

\newtheorem{conjecture}[theorem]{Conjecture}
\newtheorem{obs}[theorem]{Observation}

\title{Brooks' theorem on powers of graphs\thanks{The authors are supported by the ANR Grant EGOS (2012-2015) 12 JS02 002 01.}}

\author{
Marthe Bonamy \\
{\tt marthe.bonamy@lirmm.fr}
\and
Nicolas Bousquet \\
{\tt nicolas.bousquet@lirmm.fr}
}

\begin{document}

\maketitle

\begin{abstract}
We prove that for $k\geq 3$, the bound given by Brooks' theorem on the chromatic number of $k$-th powers of graphs of maximum degree $\Delta \geq 3$ can be lowered by $1$, even in the case of online list coloring.
\end{abstract}

\section{Introduction}
\label{sec:intro}
A graph $G=(V,E)$ is \textit{$k$-colorable} if there is a way to color each vertex with an element of $\{1,\cdots,k\}$ so that no two adjacent vertices receive distinct colors. A generalization of $k$-colorability is \textit{list $k$-colorability} (or \textit{$k$-choosability}), introduced independently by Vizing~\cite{v76} and Erd\H{o}s et al.~\cite{ert80}. The graph $G$ is $k$-choosable if for every assignment of $k$ colors to each vertex  in $V$, there is a way to color each vertex with an element of its assigned $k$ colors so that no two adjacent vertices have the same color.

Let $\Delta \geq 3$. Unless specified otherwise, the graphs considered here are simple, connected and their maximum degree is $\Delta$. Jointly with the assumption that $\Delta \geq 3$, this means for example that none of the graphs we consider is a cycle. We recall the following seminal Brooks-like theorem on choosability.

\begin{theorem}\cite{ert80}\label{th:brooks}
Except for cliques, every graph is $\Delta$-choosable.
\end{theorem}

The square $G^2$ of a graph $G=(V,E)$ is the graph obtained from $G$ by adding all edges between vertices that have  a common neighbor. Note that $\Delta(G^2)\leq \Delta^2$, so Theorem~\ref{th:brooks} implies that if $G^2$ is not a clique on $\Delta^2+1$ vertices, then $G^2$ is $\Delta^2$-choosable. In the case $\Delta=3$, Theorem~\ref{th:brooks} ensures that the square of any graph is $9$-colorable unless it is a clique. Cranston and Kim~\cite{ck08} improved this result and conjectured that it is also true for every $\Delta$. 

\begin{theorem}\cite{ck08}\label{th:cranstonkim}
Except for the Petersen graph, the square of any subcubic graph is $8$-choosable.
\end{theorem}
\emph{Moore graphs} are graphs on $\Delta^2+1$ vertices whose square is a clique~\cite{ms05}. The Petersen graph is the unique Moore graph with $\Delta=3$.
\begin{conjecture}\cite{ck08}\label{conj:cranstonkim}
Except for Moore graphs, the square of any graph is $(\Delta^2-1)$-choosable.
\end{conjecture}

The \emph{distance} between two vertices $u$ and $v$ in $G$ is the length of a shortest path between them. A generalization of the square of a graph is the $k^{th}$-power of a graph, for $k \in \mathbb{N}^*$. The \emph{$k^{th}$-power} of $G$ is obtained from $G$ by adding all edges between vertices at distance at most $k$. We denote $D(k,\Delta)$ the greatest maximum degree of a $k^{th}$-power of a graph of maximum degree $\Delta$. It basically corresponds to the maximum degree of the $k^{th}$-power of a deep enough $\Delta$-regular tree, and more precisely: \[ D(k,\Delta)=\Delta \times \sum_{i=1}^{k} (\Delta-1)^{i-1}=\Delta \times \frac{(\Delta-1)^k-1}{\Delta-2} \]
Note that $D(2,\Delta)=\Delta^2$, hence the following generalizes Conjecture~\ref{conj:cranstonkim}:
\begin{conjecture}\cite{mf13}\label{conj:gen}
For any $k \in \mathbb{N}^*$, except for Moore graphs when $k=2$, the $k^{th}$ power of any graph is $(D(k,\Delta)-1)$-choosable.
\end{conjecture}
In other words, the conjecture states that Theorem~\ref{th:brooks}, which would only yield the result for $D(k,\Delta)$, can be strengthened in the case of powers of graphs. The reason why $k=2$ is a special case is that there is no such thing as Moore graphs for higher powers (see Lemma~\ref{lem:diam} in Section~\ref{subsec:proof}).
We prove Conjecture~\ref{conj:gen} for $k \geq 3$.

\begin{theorem}\label{th:main}
For $k\geq 3$, the $k^{th}$ power of any graph is $(D(k,\Delta)-1)$-choosable. 
\end{theorem}

Independently, Conjecture~\ref{conj:gen} when $k=2$ (i.e. Conjecture~\ref{conj:cranstonkim}) has been proved recently by Cranston and Rabern~\cite{cr13}.

A generalization of list coloring, namely \emph{online list coloring}, was recently introduced independently by Schauz~\cite{s09} and Zhu~\cite{z09}. The graph $G$ is \emph{k-paintable} if, for every assignment of $k$ colors to each vertex, and for every order on the colors, there is an algorithm to color the graph by concealing until step $i$ which vertices contain the $i^{th}$ color in their lists, and deciding on the spot which vertices are colored in $i$ and which not (once colored a vertex cannot be uncolored). Clearly, online list coloring is stronger than list coloring. There exist graphs which are $k$-choosable but not $k$-paintable~\cite{s09}, though we do not know any $k$-choosable graph which is not $(k+1)$-paintable~\cite{clmptw13}. Brooks' theorem is also true in the case of online list coloring~\cite{hks10}.

We actually prove a stronger version of Theorem~\ref{th:main}, as follows.
\begin{theorem}\label{th:paint}
For $k\geq 3$, the $k^{th}$ power of any graph is $(D(k,\Delta)-1)$-paintable. 
\end{theorem}

Similarly, Cranston and Rabern proved the case $k=2$ in the more general setting of list online coloring~\cite{cr13}.

We wonder whether the following stronger generalization of Conjecture~\ref{conj:cranstonkim} could be true: 
\begin{conjecture}
For any $k \in \mathbb{N}^*$, except for a finite number of graphs, the $k^{th}$ power of any graph is $(D(k,\Delta)+1-k)$-choosable.
\end{conjecture}

\section{Proof of Theorem~\ref{th:main}}\label{sec:proof}
Let $k \geq 3$. Let $G$ be a graph. Let $M=D(k,\Delta)$. Note that $M \geq 21$ as $\Delta \geq 3$.

We will need the following lemma, which is essentially an easy adaptation of existing results (see Section~\ref{subsec:proof} for a proof).
\begin{lemma}\label{lem:glob}
If $G$ satisfies any of the following:
\begin{enumerate}
\item $G$ contains a vertex of degree smaller than $\Delta$.
\item $G$ contains a cycle shorter than $2k$.
\item $G$ contains two intersecting cycles of length $2k$.
\item diam$(G)\leq k$.
\end{enumerate}
Then $G^k$ is $(M-1)$-paintable.
\end{lemma}

Thus we can assume from now on that $G$ is $\Delta$-regular, with $g(G) \geq 2k$, diam$(G)\geq k+1$ and that the cycles of length $2k$ in $G$ are disjoint.

\begin{lemma}\label{lem:x1y1}
The graph $G$ contains two vertices $x_1$ and $y_1$ at distance $k+1$ from each other, with two neighbors $x_2,y_2$ (respectively) at distance at least $k+1$ from each other.
\end{lemma}

\begin{proof}
Since diam$(G) \geq k+1$, $G$ contains two vertices $x_1$ and $y_1$ at distance $k+1$ from each other. Let us prove that $x_1$ has a neighbor $x_2$ and $y_1$ a neighbor $y_2$ such that $x_2$ and $y_2$ are at distance at least $k+1$ from each other. Assume for contradiction that each of the $\Delta$ neighbors of $x_1$ are at distance at most $k$ from each of the $\Delta$ neighbors of $y_1$. Let $z$ be a neighbor of $x_1$. Only $\Delta-1$ neighbors of $z$ can be part of a path of length at most $k$ containing a neighbor of $y_1$, as $x_1$ is itself at distance at least $k$ from all the neighbors of $y_1$. Therefore there is a neighbor $z'$ of $z$ that belongs to two paths of length at most $k-1$ to two different neighbors of $y_1$. This yields a cycle $C$ of length at most $2k$ containing $y_1$. The cycle $C$ is actually of length $2k$ and contains $z'$, as $z'$ is the endpoint of two different paths of length at most $k$ to $y_1$ and there is no cycle of length less than $2k$ by Lemma~\ref{lem:glob}. Consequently, $y_1$ and $z'$ are diametrically opposite on $C$. Let $w$ be another neighbor of $x_1$. By the same argument, a neighbor $w'$ of $w$ belongs to a cycle $C'$ of length $2k$ that contains $y_1$, and $w'$ is diametrically opposite to $y_1$ in $C'$. Then $C$ and $C'$ intersect on $y_1$, which by Lemma~\ref{lem:glob} implies that $C$ and $C'$ are actually the same cycle. Thus $w'$ and $z'$ are actually the same vertex. Now, $(w',w,x_1,z)$ is a cycle of length $4$, a contradiction to Lemma~\ref{lem:glob} and the fact that $k \geq 3$.
\end{proof}

We will describe an algorithm to online list color $G$. Let $L$ be a list assignment of $M-1$ colors to each vertex. Since we are in the case of online list coloring, the colors will be revealed one after another (at step $1$, we learn which vertices contain color $1$ in their list, and have to decide on the spot which will be colored in it, and so on). 

At any step of the algorithm, the number of \emph{constraints} of a vertex $v$ is the number of colors in $L(v)$ that appear on vertices at distance at most $k$ from $v$. Similarly, the number of constraints \emph{implied} on a vertex $v$ by a set $S$ is the number of colors in $L(v)$ that appear on vertices of $S$. Note that the number of constraints on a vertex $v$ is bounded by its degree in $G^k$, and that this upper bound is lowered by $1$ if two neighbors of $v$ in $G^k$ have the same color or if a neighbor of $v$ in $G^k$ either is not colored or its color does not belong to $L(v)$.

We consider four vertices $x_1, x_2, y_1$ and $y_2$ obtained from Lemma~\ref{lem:x1y1}. Let $P$ be a path of length $k+1$ between $x_1$ and $y_1$. Note that by definition of $x_2,y_2$, at most one of them is on $P$. Let $v$ be a vertex at distance least two on $P$ from both $x_1$ and $y_1$ (such a vertex exists since $P$ has length at least $4$), and let $w$ be a neighbor of $v$ on $P$ distinct from $x_2$ and $y_2$. Observe that $v$ is at distance at most $k$ from all of $x_1,x_2,y_1,y_2$ and $w$ is at distance at most $k$ from $x_1,y_1$. 
Our goal is to set an order on the vertices of $G$ such that by appropriately deciding at each step whether to color or not each vertex in that order, every vertex that is considered for coloring has at most $M-2$ constraints. 
The order we choose is $x_1,x_2,y_1,y_2$, followed by all other vertices by decreasing distance to $\{v,w\}$ (the distance to a set is the minimum of the distance to each element of the set). Ties are broken arbitrarily. The order ends with $w$ and then $v$. Let us now describe more precisely the coloring algorithm. For each new color $i$:
\begin{itemize}
 \item Treat the vertices $x_1, x_2, y_1$ and $y_2$ (in a way described a little bit further).
 \item Consider all the remaining vertices, one after the other according to the chosen order. When considering a vertex $u$, color it with $i$ if $i \in L(u)$ and no neighbor of $u$ in $G^k$ is colored with $i$.
\end{itemize}
The heart of the algorithm consists in making the right decision for $\{x_1,x_2,y_1,y_2\}$ at each step, so that $v$ and $w$ each have at most $M-2$ constraints when it comes to coloring them (note that $x_1,x_2,y_1$ and $y_2$ are all at distance at most $k$ from $v$ and $w$). Let us first prove that all the other vertices are colored at the end of the algorithm.

\begin{obs}\label{obs:constraints}
Let $u$ be an uncolored vertex (distinct from $x_1,x_2,y_1,y_2$). Let $r(u)$ be the number of neighbors of $u$ in $G^k$ which appear after $u$ in the order. The number of constraints for $u$ is at most $M-r(u)$.
\end{obs}
\begin{proof}
Let $y$ be a neighbor of $u$ in $G^k$ which appears after $u$. If $y$ is not colored, then $y$ does not imply a constraint on $u$. Assume that $y$ is colored with color $i$. Since $u$ is uncolored, it means that when we tried to color $u$ with $i$, we did not succeed. So either color $i$ does not appear in $L(u)$, and then $i$ does not imply a constraint on $u$. Or another neighbor $y'$ of $u$, which appears before $u$ in the order, was colored with $i$ and then $\{y,y'\}$ implies only one constraint on $u$.
\end{proof}

Let us first prove that every vertex $u \notin \{x_1,x_2, y_1,y_2,v,w\}$ is colored at the end of the coloring algorithm (whatever the choices we did for $x_1,x_2, y_1,y_2$). Let us prove that $u$ has at most $M-2$ constraints i.e. $u$ can be colored since $|L(u)|=M-1$:
\begin{itemize}
\item If $u$ is at distance at most $k$ from both $v$ and $w$, then both $v$ and $w$ are adjacent to $u$ in $G^k$. Since they are after $u$ in the order, the result holds by Observation~\ref{obs:constraints}. 
\item If $u$ is at distance at least $k+1$ from $v$ or $w$, let $P$ be a shortest path from $u$ to $\{v,w\}$. Assume w.l.o.g. that $P$ is a shortest path from $u$ to $v$. 
Let $z_1$, $z_2$ and $z_3$ be the three vertices consecutive to $u$ in $P$. These vertices exist since $d(u,v) \geq k \geq 3$.\\
If $\{z_1,z_2,z_3\} \cap  \{x_1,x_2, y_1,y_2\}$ has size at most one, then at least two of $\{z_1,z_2,z_3\}$ are after $u$ in the order, hence the result by Observation~\ref{obs:constraints}. \\
Otherwise, at least two of  $\{z_1,z_2,z_3\}$ are in $\{x_1,x_2, y_1,y_2\}$. Since $d(x_1,y_1)\geq k+1$, if $x_1 \in \{z_1,z_2,z_3\}$ then none of $y_1,y_2$ is in this set. The same holds for $x_2$. We may assume w.l.o.g. that the intersection is exactly $x_1,x_2$. Let $w_1$ be another neighbor of $z_2$. Note that $w_1$ is neither $y_1$ nor $y_2$. Moreover $d(w_1,v)<d(u,v)$ since $P$ is a minimum path. So $w_2$ appears after $u$ in the order and $d(w_2,u) \leq k$. Two vertices at distance at most three from $u$ are after $u$ in the order, so $u$ has at most $|M|-2$ constraints.
\end{itemize}
Now, let us argue that there is a coloring of $\{x_1,x_2,y_1,y_2\}$ that ensures that $v$ and $w$ will be colored.

In standard vertex coloring, we set $x_1$ and $y_1$ to color $1$, and $x_2$ and $y_2$ to color $2$: then vertices $v$ and $w$ each have at most $M-2$ colors appearing on their neighborhood in $G^k$. So they each have at most $M-2$ constraints and then both $v$ and $w$ are colored at some step of the algorithm.

Since we are considering online list coloring, the procedure is slightly more complicated, though the idea remains the same. We want to make sure that the coloring of $\{x_1,y_1\}$ ensures that $v$ and $w$ both have one less constraint, and the coloring of $\{x_2,y_2\}$ ensures that $v$ has one less constraint. Thus when we consider $w$, it has one less constraint by $\{x_1,y_1\}$ and one less by $v$ (since $w$ is before $v$ in the order), and then $v$ has two less constraints by $\{x_1,y_1,x_2,y_2\}$. 

We proceed as follows. We denote by $NO(v)$ (resp. $NO(w)$) the number of elements of $\{x_1,y_1\}$ that are colored, minus the number of constraints implied on $v$ (resp. $w$) by elements of this set. For example, if $x_1$ and $y_1$ are colored the same, then $NO(v)=1$. The value $NO$ roughly denotes the number of colored vertices in $\{x_1,y_1\}$ which do not create a constraint. For simplicity, we consider $L(u)$ to be empty once $u$ is colored. 
At the beginning of each step $c$, we check the following:
\begin{enumerate}[(i)]
\item\label{1:same} If $c$ belongs to $L(x_1) \cap L(y_1)$, then color both $x_1$ and $y_1$ in $c$.
\item\label{1:v} If $c$ belongs to $L(x_1)$ or $L(y_1)$ but not to $L(v)$, and $NO(v)=0$, then color $x_1$ or $y_1$ in $c$.
\item\label{1:w} If $c$ belongs to $L(x_1)$ or $L(y_1)$ but not to $L(w)$, and $NO(w)=0$, then color $x_1$ or $y_1$ in $c$.
\item\label{1:end} If $c$ belongs to $L(x_1)$ or $L(y_1)$, when $M-2$ colors for the corresponding vertex have already been revealed, then color it in $c$.\newline
\item\label{2:same} If $c$ belongs to $L(x_2)\cap L(y_2)$, then color both $x_2$ and $y_2$ in $c$.
\item\label{2:v} If $c$ belongs to $L(x_2)$ or $L(y_2)$ but not to $L(v)$, then color $x_2$ or $y_2$ in $c$.
\item\label{2:end} If $c$ belongs to $L(x_2)$ or $L(y_2)$, when at least $M-4$ colors for the corresponding vertex have already been revealed, then color it in $c$.
\end{enumerate}

It remains to prove that this yields a coloring of $\{x_1,x_2,y_1,y_2\}$ such that $v$ and $w$ can be colored.

Let us first justify that $x_1$ and $y_1$ are colored in the desired way (i.e. for both $v$ and $w$, the set $\{x_1,y_1\}$ implies at most one constraint). If $x_1$ and $y_1$ are colored the same, the goal is reached. If the lists $L(x_1)$ and $L(y_1)$ have no color in common, since $|L(x_1) \cup L(y_1)|>|L(v)|$, at least one of them can be colored by~(\ref{1:v}) or~(\ref{1:w}). Then at least one of $x_1$ and $y_1$ is colored in $c \not\in L(v) \cap L(w)$, assume w.l.o.g. that $x_1$ is colored that way (and is the first if $x_1$ and $y_1$ both are). If $c \not\in L(v)\cup L(w)$, the vertex $y_1$ is never colored at~(\ref{1:v}) nor~(\ref{1:w}) ($NO(v)=NO(w)=1$), but it is colored at~(\ref{1:end}). If $c \in L(v) \cup L(w)$, assume w.l.o.g. that $c \in  L(v) \setminus L(w)$. Then, since $y_1$ was not colored before,~(\ref{1:v}) did not apply, which means that every color that belonged to $L(y_1)$ belonged to $L(v)$. But $c \in L(v) \setminus L(y_1)$ (otherwise $x_1$ and $y_1$ would be colored the same by (\ref{1:same})). Thus there remain more colors available for $y_1$ than for $v$, and~(\ref{1:v}) will eventually apply (remember that~(\ref{1:end}) does not apply before there is exactly one color left available for $y_1$).

Now, let us justify that $x_2$ and $y_2$ are colored in the desired way, i.e. if the two are colored then the set $\{x_2,y_2\}$ implies at most one constraint on $v$.
If~(\ref{2:same}) or~(\ref{2:v}) applies, the goal is reached. Let us now prove two things: that one of~(\ref{2:same}) and~(\ref{2:v}) always applies, and that both $x_2$ and $y_2$ are colored.\\
Assume that neither (\ref{2:same}) nor (\ref{2:v}) apply on $x_2$ or $y_2$. Then~(\ref{2:end}) eventually applies as only at most two colors (the colors of $x_1$ and $y_1$) may not reach (\ref{2:same})-(\ref{2:end}) and we color $x_2$ (resp. $y_2$) as soon as it has at most $2$ colors yet to be revealed (ie, if $x_2$ was not colored then it had at least $3$ colors yet to be revealed). Thus $x_2$ and $y_2$ are both colored at the end.\\
Assume that both $x_2$ and $y_2$ are colored at (\ref{2:end}). Then (\ref{2:same}) never applied, which implies $|L(x_2) \cap L(y_2)| \leq 2$. Consequently, $|L(x_2) \cup L(y_2)| \geq 2\times(M-1)-2$. Since $M\geq 21$ and $|L(v)|=M-1$, it follows that $L(x_2)\cup L(y_2)$ contains at least $18$ colors that do not belong to $L(v)$. This is a contradiction to the fact that~(\ref{2:v}) never applied.

\subsection{Proof of Lemma~\ref{lem:glob}}\label{subsec:proof}

In the following section, we prove that the different items of Lemma~\ref{lem:glob} holds. All the proofs are based on coloring algorithms based on distance, just like the proof of Theorem~\ref{th:main}. However, here we do not have to treat vertices differently: it suffices to choose an appropriate order. The resulting proofs are thus much simpler.

\begin{lemma}\label{lem:regular}
If $G$ contains a vertex of degree $\leq \Delta-1$, then $G^k$ is $(M-1)$-paintable.
\end{lemma}

\begin{proof}
Assume $G$ contains a vertex $v$ with $d(v) \leq \Delta - 1$. Since $G$ is connected, the distance to $v$ is well-defined. Order the vertices by decreasing order to $v$. At each step of the algorithm, color the vertices by decreasing distance to $v$. Every vertex $x$ at distance at least two from $v$ has at least two neighbors which are not constraints (indeed the vertices on a shortest path from $x$ to $v$ are considered after the vertex $x$ in the order). For every vertex $w$ which is a neighbor of $v$, since $k \geq 2$, the degree of $w$ in $G^k$ is at most $M-1$. Moreover, the vertex $v$ is considered after the vertex $w$ in the order, so the vertex $w$ can be colored. Since $k\geq 2$, $\Delta \geq 3$ and $d(v)\leq \Delta-1$, the degree of $v$ in $G^k$ is at most $M-\Delta < M-1$, so $v$ can be colored.
\end{proof}

Thus we assume from now on that $G$ is $\Delta$-regular.

\begin{lemma}\label{lem:shortcycle}
If $g(G) < 2k$, then $G^k$ is $(M-1)$-paintable.
\end{lemma}

\begin{proof}
Assume $G$ contains a cycle $C$ of length at most $2k-1$. Let $v$ and $w$ be two adjacent vertices on $C$. Since $C$ is of length at most $2k-1$, the degree of $v$ and $w$ in $G^k$ is less than $M-1$. Then, at each color step, we color as many vertices as possible, by decreasing distance to $\{v,w\}$ and ending with $v$ and $w$.
\end{proof}

Thus we assume from now on that $g(G)\geq 2k$.

\begin{lemma}\label{lem:shortcycles}
If $G$ contains two intersecting cycles of length $2k$, then $G^k$ is $(M-1)$-paintable.
\end{lemma}

\begin{proof}
Assume $G$ contains a vertex $v$ belonging to two cycles of length $2k$. Let $w$ be a neighbor of $v$ on one cycle of length $2k$. Vertex $v$ has degree at most $M-2$ in $G^k$, and $w$ at most $M-1$. Then, at each color step, we color as many vertices as possible, by decreasing distance to $\{v,w\}$ and ending with $w$ and then $v$.
\end{proof}

Thus we assume from now on that the cycles of length $2k$ in $G$ are disjoint.

\begin{lemma}\label{lem:diam}
If diam$(G) \leq k$ then $G^k$ is $(M-1)$-paintable. 
\end{lemma}

\begin{proof}
Assume diam$(G) \leq k$. Then $G$ contains at most $M(k,\Delta)+1$ vertices, and $G^k$ is a clique. By~\cite{ms05}, the graph $G$ contains at most $M(k,\Delta)-1$ vertices, hence the result.
\end{proof}

\section{Acknowledgements}\label{sect:ack}
The authors would like to thank Daniel Cranston for suggesting the generalization to paintability.

\bibliographystyle{plain}

\begin{thebibliography}{12}

\bibitem{clmptw13}
Carraher, J., S.~Loeb, T.~Mahoney, G.J.~Puleo, M.-T.~Tsai and D.B.~West, \emph{Three Topics in Online List Coloring}, preprint, 2013.

\bibitem{ck08}
Cranston, D. and S-J.~Kim, \emph{List-coloring the Square of a Subcubic Graph}, J. Graph Theory \textbf{57} (2008), pp.~65--87.

\bibitem{cr13}
Cranston, D. and L.~Rabern, private communication, 2013.

\bibitem{ert80}
Erd\H{o}s, P., A.L.~Rubin and H.~Taylor, \emph{Choosability in graphs}, Proc. West Coast Conf. (1980), pp.~125--157.

\bibitem{hks10}
Hladk\'y, J., D.~Kr\'al' and U.~Schauz, \emph{Brooks’ Theorem via the Alon–Tarsi Theorem}, Disc. Math. (2010), pp.~3426--3428.

\bibitem{mf13}
Miao, L.-Y. and Y.-Z.~Fan, \emph{The Distance Coloring of Graphs}, preprint, 2013, http://arxiv.org/abs/1212.1029

\bibitem{ms05}
Miller, M. and J.~\v{S}ir\'an, \emph{Moore graphs and beyond: A survey of the degree/diameter problem}, Elec. J. Comb. \textbf{61} (2005), pp.~1--63.

\bibitem{s09}
Schauz, U., \emph{Mr. Paint and Mrs. Correct}, Elec. J. Comb. \textbf{16} (2009), \#R77.

\bibitem{v76}
Vizing, V.~G., \emph{Colouring the vertices of a graph with prescribed colours (in russian)}, Diskret. Analiz \textbf{29} (1976), pp.~3--10.

\bibitem{z09}
Zhu, X., \emph{On-line list colouring of graphs}, Elec. J. Comb. \textbf{16} (2009), \#R127.
\end{thebibliography}

\end{document}